\documentclass[11pt]{llncs}
\usepackage{qtree}
\usepackage{enumerate}

\usepackage{microtype}

\newcommand*{\N}{\ensuremath{\mathbb{N}}}
\newcommand*{\nat}{\ensuremath{\mathbb{N}}}

\usepackage{latexsym}

\newcommand{\To}{\longrightarrow}
\newcommand{\manyn}[1]{{#1}_1,        \ldots       ,       {#1}_{n}}
\newcommand{\manyk}[1]{{#1}_1,        \ldots       ,       {#1}_{k}}

\renewcommand{\k}{{\sf K}}
\renewcommand{\d}{{\sf D}}
\renewcommand{\t}{{\sf T}}
\newcommand{\kf}{{\sf K4}}
\newcommand{\df}{{\sf D4}}
\newcommand{\sr}{{\sf S4}}
\newcommand{\kdfv}{{\sf KD45}}
\newcommand{\sv}{{\sf S5}}

\usepackage{complexity}

\usepackage{mathpartir}

\renewcommand{\M}{\mathcal{M}}

\newcommand{\F}{\mathcal{F}}

\usepackage{amsmath}
\usepackage{amsfonts}
\usepackage{amssymb}

\newcommand{\JL}{Justification Logic}

\title{Modal Logics with Hard Diamond-free Fragments}
\author{
	Antonis Achilleos
	} 
\institute{
	The Graduate Center of
	The City University of New York, 
	New York, USA	
	\\ 
		\email{aachilleos@gradcenter.cuny.edu
		}
}

\authorrunning{A. Achilleos}

\begin{document}
		
		\maketitle
		
\begin{abstract}
	We investigate the complexity of modal satisfiability for certain combinations of modal logics.	In particular we examine four examples of multimodal logics with dependencies and demonstrate that even if we restrict our inputs to diamond-free formulas (in negation normal form), these logics still have a high complexity. This result illustrates that having D as one or more of the combined logics, as well as the interdependencies among logics can be important sources of complexity even in the absence of diamonds and even when at the same time in our formulas we allow only one propositional variable. We then further investigate and characterize the complexity of the diamond-free, 1-variable fragments of multimodal logics in a general setting.
	
\end{abstract}

			\section{Introduction}
			
			The complexity of the satisfiability problem for modal logic, and thus of its dual, modal provability/validity, has been extensively studied. Whether one is interested in areas of application of Modal Logic, or in the properties of Modal Logic itself, the complexity of modal satisfiability plays an important role. Ladner has established most of what are now considered classical results on the matter \cite{ladnermodcomp}, determining that most of the usual modal logics are \PSPACE-hard, while more for the most well-known logic with negative introspection, \sv, satisfiability is \NP-complete; Halpern and Moses \cite{130913} then demonstrated that \kdfv-satisfiability is \NP-complete and that the multi-modal versions of these logics are \PSPACE-complete. Therefore, it makes sense to try to find fragments of these logics that have an easier satisfiability problem by restricting the modal elements of a formula -- or prove that satisfiability remains hard even in fragments that seem trivial (ex. \cite{Halpern95theeffect,Chagrov02howmany}). 
			In this paper we present mostly hardness results for this direction and for certain cases of multimodal logics with  modalities that affect each other. 
			Relevant syntactic restrictions and their effects on the complexity of various modal logics have been examined in \cite{hemaspaandra2001PoorMan} and \cite{Hemaspaandra2010GeneralizedModalSat}.
			For more  on Modal Logic and its complexity, see \cite{130913,Fagin1995ReasoningAboutKnowledge,Spaan93PhD}.
			
			A (uni)modal formula is a formula formed by using propositional variables and Boolean connectives, much like propositional calculus, but we also use two additional operators, $\Box$ (box) and $\Diamond$ (diamond):  if $\phi$ is a formula, then $\Box \phi$ and $\Diamond \phi$ are formulas. Modal formulas are given truth values with respect to a Kripke model $(W,R,V)$,\footnote{There are numerous semantics for modal logic, but in this paper we only use Kripke semantics.} which can be seen as a directed graph $(W,R)$ (with possibly an infinite number of vertices and allowing self-loops) together with a truth value assignment for the propositional variables for each world (vertex) in $W$, called $V$. We define $\Box \phi$ to be true in a world $a$ if $\phi$ is true at every world $b$ such that $(a,b)$ is an edge, while $\Diamond$ is the dual operator: $\Diamond \phi$ is true at $a$ if $\phi$ is true at some $b$ such that $(a,b)$ is an edge. 
			
			We are interested in the complexity of the satisfiability problem for modal formulas (in negation normal form, to be defined later) that have no diamonds -- i.e. is there a model with a world at which our formula is true? When testing a modal formula for satisfiability (for example, trying to construct a model for the formula through a tableau procedure), a clear source of complexity are the diamonds in the formula. When we try to satisfy $\Diamond \phi$, we need to assume the existence of an extra world where $\phi$ is satisfied. When trying to satisfy $\Diamond p_1 \wedge \Diamond p_2 \wedge \Box \phi_n$, we require two new worlds where $p_1\wedge \phi_n$ and $p_2 \wedge \phi_n$ are respectively satisfied; for example, for $\phi_0 = \top$ and $\phi_{n+1} = \Diamond p_1 \wedge \Diamond p_2 \wedge \Box \phi_n$, this causes an exponential explosion to the size of the constructed model (if the model we construct for $\phi_n$ has $k$ states, then the model for $\phi_{n+1}$ has $2k+1$ states). There are several modal logics, but it is usually the case that in the process of satisfiability testing, as long as there are no diamonds in the formula, we are not required to add more than one world to the constructed model. Therefore, it is natural to identify the existence of diamonds as an important source of complexity. On the other hand, when the modal logic is \d,  its models are required to have a serial accessibility relation (no sinks in the graph). Thus, when we test $\Box \phi$ for \d-satisfiability, we require a world where $\phi$ is satisfied. In such a unimodal setting and in the absence of diamonds, we avoid an exponential explosion in the number of worlds and we can consider models with only a polynomial number of worlds.
			
			Several authors have examined the complexity of combinations of modal logic (ex. \cite{marx1997multi,gabbay2003many,kurucz200715}). Very relevant to this paper work on the complexity of combinations of modal logic is by
			Spaan in \cite{Spaan93PhD} and Demri in \cite{DBLP:conf/tableaux/Demri00}. In particular, Demri studied $L_1 \oplus_\subseteq L_2$, which is $L_1 \oplus L_2$ (see \cite{Spaan93PhD}) with the additional axiom $\Box_2\phi \rightarrow \Box_1\phi$ and where $L_1, L_2$ are among {\k}, {\sf T, B}, \sr, and \sv\ -- modality 1 comes from $L_1$ and 2 from $L_2$. For when $L_1$ is among {\sf K, T, B} and $L_2$ among {\sr, \sv}, he establishes \EXP-hardness for $L_1 \oplus_\subseteq L_2$-satisfiability. We consider $L_1 \oplus_\subseteq L_2$, where $L_1$ is a unimodal or bimodal logic (usually \d, or \df). When $L_1$ is bimodal, $L_1 \oplus_\subseteq L_2$ is $L_1 \oplus L_2$ with the extra axioms $\Box_3\phi \rightarrow \Box_1\phi$ and $\Box_3\phi \rightarrow \Box_2\phi$.
			
			The family of logics we consider in this paper can be considered part of the much more general family of \emph{regular grammar logics (with converse)}. Demri and De Nivelle have shown in \cite{DemriDeNivelle2005decidingregulargram} through a translation into a fragment of first-order logic that the satisfiability problem for the whole family is in \EXP\ (see also \cite{Demri01122001}). Then, Nguyen and Sza{\l}as in \cite{nguyen2011exptime}  gave a tableau procedure for the general satisfiability problem (where the logic itself is given as input in the form of a finite automaton) and determined that it is also in \EXP. 
			
			In this paper, we examine the effect on the complexity of modal satisfiability testing of restricting our input to diamond-free formulas under the requirement of seriality and in a multimodal setting with connected modalities. In particular, we initially examine four examples: $\d_2 \oplus_\subseteq \k$, $\d_2 \oplus_\subseteq \kf$, $\d \oplus_\subseteq \kf$, and $\df_2 \oplus_\subseteq \kf$.\footnote{In general, in $A\oplus_\subseteq B$, if $A$ a bimodal (resp. unimodal) logic, the modalities 1 and 2 (resp. modality 1) come(s) from $A$ and 3 (resp. 2) comes from logic $B$.} For these logics we look at their diamond-free fragment and establish that they are \PSPACE-hard and in the case of $\d_2 \oplus_\subseteq \kf$, \EXP-hard. Furthermore, $\d_2 \oplus_\subseteq \k$, $\d \oplus_\subseteq \kf$, and $\df_2 \oplus_\subseteq \kf$ are \PSPACE-hard  and $\d_2 \oplus_\subseteq \kf$ is \EXP-hard even for their 1-variable fragments.
			Of course these results can be naturally extended to more modal logics, but we treat what we consider  simple characteristic cases. For example, it is not hard to see that nothing changes when in the above multimodal logics we replace {\k} by {\d}, or {\kf} by {\df}, as the extra axiom $\Box_3 \phi \rightarrow \Diamond_3 \phi$ ($\Box_2 \phi \rightarrow \Diamond_2 \phi$ for $\d \oplus_\subseteq \kf$) is a derived one. It is also the case that in these logics we can replace $\kf$ by other logics with positive introspection (ex. \sr, \sv) without changing much in our reasoning.
			
			Then, we examine a general setting of a multimodal logic (we consider combinations of modal logics \k, \d, \t, \df, \sr, \kdfv, \sv) where we include axioms $\Box_i\phi \rightarrow \Box_j\phi$ for some pairs $i,j$. For this setting we 
			determine exactly 
			the complexity of satisfiability for the diamond-free (and 1-variable) fragment of the logic and we are able to make some interesting observations. The study of this general setting is of interest, because determining exactly when the complexity drops to tractable levels for the diamond-free fragments illuminates possibly appropriate candidates for parameterization: if the complexity of the diamond-free, 1-variable fragment of a logic drops to \P, then we may be able to develop algorithms for the satisfiability problem of the logic that are efficient for formulas of few diamonds and propositional variables; if the complexity of that fragment does not drop, then the development of such algorithms seems unlikely (we may be able to parameterize with respect to some other parameter, though). Another argument for the interest of these fragments results from the hardness results of this paper. The fact that the complexity of the diamond-free, 1-variable fragment of a logic remains high means that this logic is likely a very expressive one, even when deprived of a significant part of its syntax.
			
			A very relevant approach is presented in \cite{hemaspaandra2001PoorMan,Hemaspaandra2010GeneralizedModalSat}. In \cite{hemaspaandra2001PoorMan}, Hemaspaandra determines the complexity of Modal Logic when we restrict the syntax of the formulas to use only a certain set of operators. In \cite{Hemaspaandra2010GeneralizedModalSat}, Hemaspaandra et al. consider multimodal logics and all Boolean functions. In fact, some of the cases we consider have already been studied in \cite{Hemaspaandra2010GeneralizedModalSat}. Unlike \cite{Hemaspaandra2010GeneralizedModalSat}, we focus on multimodal logics where the modalities are not completely independent -- they affect each other through axioms of the form  $\Box_i\phi \rightarrow \Box_j\phi$. Furthermore in this setting we only consider diamond-free formulas, while at the same time we examine the cases where we allow only one propositional variable. As far as our results are concerned, it is interesting to note that in \cite{hemaspaandra2001PoorMan,Hemaspaandra2010GeneralizedModalSat} when we consider frames with serial accessibility relations, the complexity of the logics under study tends to drop, while in this paper we see that serial accessibility relations (in contrast to arbitrary, and sometimes reflexive, accessibility relations) contribute substantially to the complexity of satisfiability.
			
			Another motivation we have  is the relation between the diamond-free fragments of Modal Logic with Justification Logic. Justification Logic can be considered an explicit counterpart of Modal Logic. It introduces justifications to the modal language, replacing boxes ($\Box$) by constructs called justification terms. When we examine a justification formula with respect to its satisfiability, the process is  similar to examining the satisfiability of a modal formula without any diamonds (with some extra nontrivial parts to account for the justification terms). Therefore, as we are interested in the complexity of systems of Multimodal and Multijustification Logics, we are also interested in these diamond-free fragments. For more on Justification Logic and its complexity, the reader can see \cite{Art08RSL,Kuz08PhD}; for more on the complexity of Multi-agent \JL\ and how this paper is connected to it, the reader can see \cite{AchilleosDissertation}.
			
			It may seem strange that we restrict ourselves to formulas without diamonds but then we implicitly reintroduce diamonds to our formulas by considering serial modal logics -- still, this is not the same situation as allowing the formula to have any number of diamonds, as seriality is only responsible for introducing at most one accessible world (for every serial modality) from any other. This is a nontrivial restriction, though, as we can see from this paper's results. Furthermore it corresponds well with the way justification formulas behave when tested for satisfiability.
			
			
			
			\section{Modal Logics
			and Satisfiability}
			
			For the purposes of this paper it is convenient to consider modal formulas in negation normal form (NNF) -- negations are pushed to the atomic level (to the propositional variables) and we have no implications
			 -- and this is the way we define our languages.
			Note that for all logics we consider, every formula can be converted easily to its NNF form, so the NNF fragment of each logic we consider has exactly the same complexity as the full logic.
			We discuss modal logics with one, two, and three modalities, so we have three modal languages, $L_1 \subseteq L_2 \subseteq L_3$. They all include propositional variables, usually called $p_1,p_2,\ldots$ (but this may vary based on convenience) and $\bot$. If $p$ is a propositional variable, then $p$ and $\neg p$ are called  literals and are also included in the language and so is $\neg \bot$, usually called $\top$. If $\phi, \psi$ are in one of these languages, so are $\phi \vee \psi$ and $\phi \wedge \psi$. Finally, if $\phi$ is in $L_3$,
			then so are 
			$\Box_1\phi,\Box_2\phi,\Diamond_1\phi,\Diamond_2\phi,\Box_3\phi,\Diamond_3\phi$.
			$L_2$ includes all formulas in $L_3$ that have no $\Box_3, \Diamond_3$ and $L_1$ includes all formulas in $L_2$ that have no $\Box_2, \Diamond_2$.
			In short, $L_n$ is defined in the following way, where $1\leq i \leq n$:\ \
			$
			\phi ::= p\ |\ \neg p\ |\ \bot \ |\ \neg \bot \ |\ \phi \land \phi\ |\ \phi \lor \phi\ |\ \Diamond_i\phi\ |\ \Box_i \phi
			.$
%
			If we did not only consider formulas in negation normal form, we would include $\neg \phi$.  
			If we consider formulas in $L_1$, $\Box_1$ may just be called $\Box$.\footnote{It may seem strange that we introduce languages with diamonds and then only consider their diamond-free fragments. When we discuss \k, we consider the full language, so we introduce diamonds for $L_1,L_2,L_3$ for uniformity.}
			
			A Kripke model for a trimodal logic (a logic based on language $L_3$) is a tuple $\M = (W,R_1,R_2,R_3,V)$, where $R_1,R_2,R_3 \subseteq W\times W$ and for every propositional variable $p$, $V(p)\subseteq W$. Then, $(W,R_1,V)$ (resp. $(W,R_1,R_2,V)$) is a Kripke model for a unimodal (resp. bimodal) logic. Then, $(W,R_1), (W,R_1,R_2)$, and $(W,R_1,R_2,R_3)$ are called frames and $R_1,R_2,R_3$ are called accessibility relations. We  define the truth relation $\models$ between models, worlds (elements of $W$, also called states) and formulas in the following recursive way:
			\begin{itemize}
				\item[]
				$\M,a \not \models \bot$;
				\item[]
				$\M,a \models p$ iff $a\in V(p)$ and 	
				$\M,a \models \neg p$ iff $a \notin V(p)$;
				\item[]
				$\M,a \models \phi \wedge \psi$ iff both $\M,a \models \phi$ and $\M, a \models \psi$;
				\item[]
				$\M,a \models \phi \vee \psi$ iff  $\M,a \models \phi$ or $\M, a \models \psi$;
				\item[]
				$\M,a \models \Diamond_i \phi$ iff there is some $b \in W$ such that $a R_i b$ and $\M,b \models \phi$;
				\item[]
				$\M,a \models \Box_i \phi $ iff for all $b \in W$ such that $a R_i b$ it is the case that $\M,b \models \phi$.
			\end{itemize}

			In this paper we deal with five logics: \k, $\d_2 \oplus_\subseteq \k$, $\d_2 \oplus_\subseteq \kf$, $\d \oplus_\subseteq \kf$, and $\df_2 \oplus_\subseteq \kf$. All except for {\k} and $\d \oplus_\subseteq \kf$ are trimodal logics, based on language $L_3$, $\k$ is a unimodal logic (the simplest normal modal logic) based on $L_1$, and $\d \oplus_\subseteq \kf$ is a bimodal logic based on $L_2$. Each modal logic $M$ is associated with a class of frames $C$. A formula $\phi$ is then called $M$-satisfiable iff there is a frame $\F \in C$, where $C$ the class of frames associated to $M$, a model $\M = (\F,V)$, and a state $a$ of $\M$ such that $\M,a \models \phi$. We  say that $\M$ satisfies $\phi$, or $a$ satisfies $\phi$ in $\M$, or  $\M$ models $\phi$, or that $\phi$ is true at $a$. 
			\begin{description}
				\item[\k] 
				is the logic associated with the class of all frames; 
				\item[$\d_2 \oplus_\subseteq \k$] 
				is  the logic associated with the class of frames $\F = (W,R_1,R_2,R_3)$ for which $R_1,R_2$ are serial (for every $a$ there are $b, c$ such that $a R_1b$, $a R_2 c$) and $R_1 \cup R_2 \subseteq R_3$; 
				\item[$\d_2 \oplus_\subseteq \kf$] 
				is the logic associated with the class of frames $\F = (W,R_1,R_2,R_3)$ for which $R_1,R_2$ are serial, $R_1 \cup R_2 \subseteq R_3$, and $R_3$ is transitive; 
				\item[$\d \oplus_\subseteq \kf$] 
				is the logic associated with the class of frames $\F = (W,R_1,R_2)$ for which $R_1$ is serial, $R_1 \subseteq R_2$, and $R_2$ is transitive; 
				\item[$\df_2 \oplus_\subseteq \kf$] 
				is the logic associated with the class of frames $\F = (W,R_1,R_2,R_3)$ for which $R_1,R_2$ are serial, $R_1 \cup R_2 \subseteq R_3$ and $R_1,R_2,R_3$ are transitive.
			\end{description}
			
			\subsubsection*{Tableau}
			
			A way to test for satisfiability is by using a tableau procedure. A good source on tableaux is \cite{d1999handbook}.
			We present tableau rules for {\k} and for the diamond-free fragments of $\d_2 \oplus_\subseteq \k$ and then for the remaining three logics. The main reason we present these rules is because they are useful for later proofs and because they help to give intuition regarding the way we can test for satisfiability.
			The ones for {\k} are classical and follow right away. Formulas used in the tableau are given a prefix, which intuitively corresponds to a state in a model we attempt to construct and is a string of natural numbers, with $.$ representing concatenation. The tableau procedure for a formula $\phi$ starts from $0\ \phi$ and applies the rules it can to produce new formulas and add them to the set of formulas we construct, called a branch. A rule of the form $\frac{a}{b\ |\ c}$ means that the procedure nondeterministically chooses between $b$ and $c$ to produce, i.e. a branch is closed under that application of that rule as long as it includes $b$ or $c$. If the branch has $\sigma\ \bot$, or both $\sigma\ p$ and $\sigma\ \neg p$, then it is called propositionally closed and the procedure rejects its input. Otherwise, if the branch contains $0\ \phi$, is closed under the rules, and is not propositionally closed, it is an accepting branch for $\phi$; the procedure accepts $\phi$ exactly when there is an accepting branch for $\phi$. The rules for {\k} are in Table \ref{table:tableauforK}.
			
			\begin{table}[t]
			\begin{minipage}[l][18ex]{0.2\linewidth}
				\[ 
				\inferrule{\sigma\ \phi \vee \psi}{\sigma\  \phi \ \mid\ \sigma\  \psi }  
				\]
			\end{minipage}\hfill
			\begin{minipage}[l][18ex]{0.2\linewidth}
				\[ 
				\inferrule{\sigma\ \phi \wedge \psi}{\sigma\  \phi \\\\ \sigma\  \psi }  
				\]
			\end{minipage}\hfill
			\begin{minipage}[l][18ex]{0.22\linewidth}
				\[ 
				\inferrule{\sigma\ \Box\phi}{\sigma.i\ \phi }  
				\] where $\sigma.i$ has already appeared in the branch.
			\end{minipage}\hfill
			\begin{minipage}[l][18ex]{0.2\linewidth}
				\[ 
				\inferrule{\sigma\ \Diamond\phi}{\sigma.i\ \phi }  
				\] where $\sigma.i$ has not yet appeared in the branch.
			\end{minipage}
			\caption{Tableau rules for \k.}
						\label{table:tableauforK}
		\end{table}
			
			For the remaining logics, we are only concerned with their diamond-free fragments, so we only present rules for those to make things simpler. As we mention in the Introduction, all the  logics we consider can be seen as regular grammar logics with converse (\cite{DemriDeNivelle2005decidingregulargram}), for which the satisfiability problem is in \EXP. This already gives an upper bound for the satisfiability of $\d_2 \oplus_\subseteq \kf$ (and for the general case of $(N,\subset,F)$ from Section \ref{sec:general}). We present the tableau rules anyway (without proof), since it helps to visually give an intuition of each logic's behavior, while it helps us reason about how some logics reduce to others.
			
			To give some intuition on the tableau rules, the main differences from the rules for \k\ are that in a frame for these logics we have two or three different accessibility relations (lets assume for the moment that they are $R_1,R_3$, and possibly $R_2$), that one of them ($R_3$) is the (transitive closure of the) union of the others, and that we can assume that due to the lack of diamonds and seriality, $R_1$ and $R_2$ are total functions on the states. To establish this, notice that the truth of diamond-free formulas in NNF is preserved in submodels; when $R_1, R_2$ are not transitive, we can simply keep removing pairs from $R_1,R_2$ in a model as long as they remain serial. As for the tableau for $\df_2 \oplus_\subseteq \kf$, notice that for $i = 1,2$, $R_i$
			 can map each state $a$ to some $c$ such that for every $\Box_i \psi$, subformula of $\phi$, $c \models \Box_i\psi \rightarrow \psi$. If $a$ is such a $c$, we map $a$ to $a$; otherwise we can find such a $c$ in the following way. Consider a sequence $bR_i c_1 R_i c_2 R_i \cdots$; if some $c_j \not \models \Box_i\psi \rightarrow \psi$, then $c_j \models \Box_i\psi$, so for every $j'>j$, $c_{j'} \models \Box_i\psi \rightarrow \psi$. Since the subformulas of $\phi$ are finite in number, we can find some large enough $j \in \nat $ and set $c = c_j$.
			 Notice that using this construction on $c$, $R_i$ maps $c$ to $c$, is transitive and serial.
			
			The rules for $\d_2 \oplus_\subseteq \k$ are in Table \ref{table:tableauforD2K}.
			\begin{table}[t]\noindent
			\begin{minipage}[l][14ex]{0.22\linewidth}
				\[ 
				\inferrule{\sigma\ \phi \vee \psi}{\sigma\  \phi \ \mid\ \sigma\  \psi }  
				\]
			\end{minipage}\hfill
			\begin{minipage}[l][14ex]{0.2\linewidth}
				\[ 
				\inferrule{\sigma\ \phi \wedge \psi}{\sigma\  \phi \\\\ \sigma\  \psi }  
				\]
			\end{minipage}\hfill
			\begin{minipage}[l][12ex]{0.19\linewidth}
				\[ 
				\inferrule{\sigma\ \Box_1\phi}{\sigma.1\ \phi }  
				\] 
			\end{minipage}\hfill
			\begin{minipage}[l][12ex]{0.19\linewidth}
				\[ 
				\inferrule{\sigma\ \Box_2\phi}{\sigma.2\ \phi }  
				\] 
			\end{minipage}\hfill
			\begin{minipage}[l][12ex]{0.19\linewidth}
				\[ 
				\inferrule{\sigma\ \Box_3\phi}{\sigma.1\ \phi \\\\ \sigma.2\ \phi}  
				\] 
			\end{minipage}\hfill
			\caption{The rules for $\d_2 \oplus_\subseteq \k$}
			\label{table:tableauforD2K}
		\end{table}

			We sketch a proof that these tableau procedures are correct, i.e. for every diamond-free $\phi$,
			there is a model for $\phi$ iff there is an accepting branch for $\phi$. From an accepting branch for $\phi$ we construct a model for $\phi$: let $W$ be all the prefixes that have appeared in the branch, 
			$$R_1 = \{ (w,w.1) \in W^2 \} \cup \{(w,w)\in W^2\mid w.1 \notin W \},$$ 
			$$R_2 = \{ (w,w.2) \in W^2 \}\cup \{(w,w)\in W^2\mid w.2 \notin W \},$$ 
			$R_3 = R_1 \cup R_2
			,$ and $V(p) = \{w \in W \mid  w\ p $ appears in the branch$ \}$. Then, it is not hard to see that $(W,R_1,R_2,R_3)$ is indeed a frame for $\d_2 \oplus_\subseteq \k$ ($R_1,R_2 \subseteq R_3$ and they are all serial), and that for $\M = (W,R_1,R_2,R_3,V)$, $\M, 0 \models \phi$ -- by proving through  a straightforward induction on $\psi$ that for every $w\ \psi$ in the branch, $\M,w \models \psi$. 
			
			On the other hand, given  $\M,a \models \phi$, we can construct an accepting branch for $\phi$ in the following way. We map $0$ to $a$ and for every $w.i$, where $i = 1,2$, if $w$ is mapped to state $b$ of the model, then $w.i$ is mapped to some state $c$, where $b R_i c$. Then we can make sure we make appropriate nondeterministic choices when applying a rule to ensure that whenever $w\ \psi$ is produced and $w$ is mapped to $a$, then  $\M,a \models \psi$: if $\psi = \phi$, then this is trivially correct; if we apply the first rule on $w\ \psi_1\vee \psi_2$, then since $\M,a \models\psi_1\vee \psi_2$, it is the case that $\M,a \models \psi_1$ or $\M,a \models \psi_2$ and we can choose the appropriate formula to introduce to the branch; the remaining rules are trivial. Therefore, the branch can never be propositionally closed.
			
		To come up with tableau rules for the other three logics, we can modify the above rules. The first two rules that cover the propositional cases are always the same, so we give the remaining rules for each case. In the following, notice that the resulting branch may be infinite. However we can simulate such an infinite branch by a finite one: we can limit the size of the prefixes, as after a certain size (up to $2^{|\phi|}$, where $\phi$ the tested formula) it is guaranteed that there will be two prefixes that prefix the exact same set of formulas. Thus, we can either assume the procedure terminates or that it generates a full branch, depending on our needs. In that latter case, to ensure a full branch is generated, we can give lowest priority to a rule when it generates a new prefix.

			The rules for the diamond-free fragment of $\d_2 \oplus_\subseteq \kf$ are in Table \ref{table:d2_kf};
			\begin{table}[t]
			\begin{minipage}[l][12ex]{0.25\linewidth}
				\[ 
				\inferrule*{\sigma\ \Box_1\phi}{\sigma.1\ \phi }  
				\] 
			\end{minipage}\hfill
			\begin{minipage}[l][12ex]{0.25\linewidth}
				\[ 
				\inferrule*{\sigma\ \Box_2\phi}{\sigma.2\ \phi }  
				\] 
			\end{minipage}\hfill
			\begin{minipage}[l][17ex]{0.25\linewidth}
				\[ 
				\inferrule*{\sigma\ \Box_3\phi}{\sigma.1\ \phi \\\\ \sigma.2\ \phi \\\\ \sigma.1\ \Box_3\phi \\\\ \sigma.2\ \Box_3 \phi}  
				\] 
			\end{minipage}\hfill
			\caption{Tableau rules for the diamond-free fragment of $\d_2 \oplus_\subseteq \kf$}
						\label{table:d2_kf}
		\end{table}
			the rules for the diamond-free fragment of 
			$\d \oplus_\subseteq \kf$  in Table \ref{table:d_kf};
%
			\begin{table}[t]
			\begin{minipage}[l][10ex]{0.22\linewidth}
				\[ 
				\inferrule*{\sigma\ \Box_1\phi}{\sigma.1\ \phi }  
				\] 
			\end{minipage}\hfill
			\begin{minipage}[l][10ex]{0.23\linewidth}
				\[ 
				\inferrule*{\sigma\ \Box_2\phi}{\sigma.1\ \phi \\\\ \sigma.1\ \Box_2 \phi }  
				\] 
			\end{minipage}\hfill 
			\hspace{0.22\linewidth}
			\caption{Tableau rules for the diamond-free fragment of $\d \oplus_\subseteq \kf$}
						\label{table:d_kf}
		\end{table}
			and the rules for the diamond-free fragment of 
			$\df_2 \oplus_\subseteq \kf$ are in Table \ref{table:df2_kf}.
			\begin{table}[t]
			\begin{minipage}[l][12ex]{0.29\linewidth}
				\[ 
				\inferrule*{\sigma\ \Box_1\phi}{n_1(\sigma)\ \phi 
				}  
				\] 
			\end{minipage}\hfill
			\begin{minipage}[l][12ex]{0.29\linewidth}
				\[ 
				\inferrule*{\sigma\ \Box_2\phi}{n_2(\sigma)\ \phi 
				}  
				\] 
			\end{minipage}\hfill
			\begin{minipage}[l][15ex]{0.30\linewidth}
				\[ 
				\inferrule*{\sigma\ \Box_3\phi}{n_1(\sigma)\ \phi \\\\ n_2(\sigma)\ \phi \\\\ n_1(\sigma)\ \Box_3\phi \\\\ n_2(\sigma)\ \Box_3\phi}  
				\] 
			\end{minipage}\hfill\\
			where $n_i(\sigma) = \sigma$ if $\sigma = \sigma'.i$ for some $\sigma'$ and $n_i(\sigma)=\sigma.i$ otherwise.
			\caption{Tableau rules for the diamond-free fragment of $\df_2 \oplus_\subseteq \kf$}
			\label{table:df2_kf}
		\end{table}

We skip any proof of correctness for these cases, as they are similar to the previous case. The exception is the tableau procedure for $\df_2 \oplus_\subseteq \kf$, which is a little different and for which we must give some adjustments in the constructions of the model from the accepting branch and of the accepting branch from a model. The construction of the model is similar as for the case of $\d_2\oplus_\subseteq \k$, only this time for $i=\{1,2\}$ $R_i = \{(\sigma,n_i(\sigma))\in W^2 \}\cup \{(\sigma,\sigma)\in W^2\mid n_i(\sigma) \notin W \}$ (notice they are transitive) and $R_3$ the transitive closure of $R_1\cup R_2$. On the other hand, when constructing an accepting branch, we need to make sure that if we map $\sigma$ to $b$, then we map $\sigma.i$ to some $c$ such that for every $\Box_i \psi$, subformula of $\phi$, $c \models \Box_i\psi \rightarrow \psi$. We can find such a $c$ by considering a sequence $bR_i c_1 R_i c_2 R_i \cdots$; if some $c_j \not \models \Box_i\psi \rightarrow \psi$, then $c_j \models \Box_i\psi$, so for every $j'>j$, $c_j \models \Box_i\psi \rightarrow \psi$. Since the subformulas of $\phi$ are finite in number, we can find some large enough $j \in \nat $ and set $c = c_j$.


\begin{proposition}\label{prp:upper}
	The satisfiability problem for the diamond-free fragments of $\d_2 \oplus_\subseteq \k$, of $\d \oplus_\subseteq \kf$, and of $\df_2 \oplus_\subseteq \kf$ is in \PSPACE;  satisfiability  for the diamond-free fragment of $\d_2 \oplus_\subseteq \kf$ is in \EXP.
\end{proposition}
\begin{proof} 
	We can use the rules to prove that  satisfiability of the diamond-free fragment of $\d_2 \oplus_\subseteq \k$ is in \PSPACE. In fact, we can use an alternating polynomial-time algorithm to simulate the tableau procedure and given a formula $\phi$ to construct an accepting branch for $\phi$. The algorithm uses an existential non-deterministic choice when we apply the first rule to choose which of the resulting prefixed formulas to add to the branch; it also uses a universal choice to choose between $\sigma.1$ and $\sigma.2$ for every $\sigma$ it has produced. Other than that, it applies all the tableau rules it can, until there are none left.  It is not hard to construct an accepting tableau branch from an accepting run of the algorithm and vice-versa. The fact that the algorithm runs in polynomial time can be established by observing that only up to $|\phi|$ formulas can be prefixed by a specific prefix, while the nesting depth of the boxes in the formulas (also called modal depth) strictly decreases as the length of their prefix increases.
	
	To establish upper complexity bounds for the diamond-free fragments of the remaining logics, we can use a similar procedure, only this time it is an alternating polynomial \emph{space} algorithm to simulate the tableau procedure -- we do not have the same bounds on the length of the prefixes as above, but we can just keep formulas prefixed by a single prefix in memory and as we argued before this is at most $|\phi|$ formulas -- of course this means we give priority to propositional rules. Furthermore we do not even need to keep the current prefix in memory, but we can just use a counter of polynomial size for the length of the prefix (an important point, because the length of a prefix can be exponential); when the counter becomes larger than $2^{|\phi|}$, then of course we can terminate. This gives an ({\sf A}$\PSPACE =$)\EXP-upper bound for the complexity of satisfiability for the diamond-free fragment of $\d_2 \oplus_\subseteq \kf$; to get a \PSPACE-upper bound for the other two logics, notice that  the tableau for $\d \oplus_\subseteq \kf$ uses only prefixes of the form $0.1^x$ and the tableau for $\df_2 \oplus_\subseteq \kf$ only subprefixes of  $0.(1.2)^\omega$ and $0.(2.1)^\omega$, therefore making  universal choices unnecessary.
				\qed
			\end{proof}
			
			The cases of $\d \oplus_\subseteq \kf$ and $\df_2 \oplus_\subseteq \kf$ are especially interesting.  In \cite{DBLP:conf/tableaux/Demri00}, Demri established that $\d \oplus_\subseteq \kf$-satisfiability (and because of the following section's results also $\df_2 \oplus_\subseteq \kf$-satisfiability) is \EXP-complete. In this paper, though, we establish that the complexity of these two logics' diamond-free (and one-variable) fragments are \PSPACE-complete (in this section we establish the \PSPACE\ upper bounds, while in the next one the lower bounds), which is a drop in complexity (assuming $\PSPACE \neq \EXP$), but not one that makes the problem tractable (assuming $\P \neq \PSPACE$).
			
			\section{Lower Complexity Bounds}
			
			In this section we give hardness results for the logics of the previous section -- except for \k. In \cite{Chagrov02howmany}, the authors prove that the variable-free fragment of {\k} remains \PSPACE-hard. We make use of that result here and prove the same for the diamond-free, 1-variable fragment of 	$\d_2 \oplus_\subseteq \k$. Then we prove \EXP-hardness for the diamond-free fragment of $\d_2 \oplus_\subseteq \kf$ and \PSPACE-hardness for the diamond-free fragments of $\d \oplus_\subseteq \kf$ and of $\df_2 \oplus_\subseteq \kf$, which we later improve to the same result for the diamond-free, 1-variable fragments of these logics.
			
			\begin{proposition}\label{prp:PSPford2k}
				The diamond-free, 1-variable fragment of $\d_2 \oplus_\subseteq \K$ is \PSPACE-complete.
			\end{proposition}
			\begin{proof} 
				The upper bound was given by Proposition \ref{prp:upper}.
				We  give a translation from unimodal formulas to formulas of three modalities such that $\phi$ is \K-satisfiable if and only if $\phi^{tr}$ (the result of the translation) is $\d_2 \oplus_\subseteq \k$-satisfiable. The translation uses an extra propositional variable (not appearing in $\phi$), $q$. It is defined in the following way.
				
				We want the tableau for $\phi^{tr}$ to simulate the tableau for $\phi$. However, $\phi$ may have diamonds, which are not allowed in $\phi^{tr}$. When the tableau for \k\ encounters a diamond, then it generates a unique prefix. Therefore, we must replace a diamond with something which will generate a unique prefix in the tableau for $\d_2 \oplus_\subseteq \k$. This unique prefix can be generated by a unique sequence of boxes, which is provided by function $dseq$ (defined below):
				
				For a formula $\phi$, let $\manyk{\theta}$ be an enumeration of its subformulas 
				and in increasing order with respect to their size (to ensure that if $\eta_1$ is a subformula of $\eta_2$, then $\eta_1$ appears first). Also, 
				let\footnote{
					Notice that if there is at least one diamond in $\phi$, then $\phi$ has at least two subformulas, thus if there are diamonds, then $\log k \geq 1$; if $k=1$, then this discussion is meaningless: $\phi^{tr} = \phi$.
					} 
					\[dseq: \{1,2,\ldots, k \} \To \{\Box_1,\Box_2\}^{\lceil \log k \rceil }\] 
				be some one-to-one mapping from those subformulas to a unique sequence of boxes. The actual mapping is not important, but an easy choice would be  $dseq(x) = \Box_{x_1 + 1}\Box_{x_2+1}\cdots \Box_{x_{\lceil \log k \rceil } +1}$, where $bin(x) := x_1x_2\cdots x_{\lceil \log k \rceil }$ is the binary representation of $x$ -- so this is the one we assume. 
				We can define $i^{tr}$ by recursion	 on $i$: 
				\begin{itemize}
					\item if $\theta_i$ is a literal, $\top$, or $\bot$, then
					$i^{tr} = \theta_i$; 
					\item
					if $\theta_i = \theta_j \circ \theta_l$, where $\circ$ is either $\wedge$ or $\vee$, then
					$i^{tr} = j^{tr} \circ l^{tr}$; 
					\item
					if $\theta_i = \Box \theta_j$, then
					$i^{tr} = \Box_3^{\lceil \log k \rceil} ( j^{tr} \vee \neg q ) $; 
					\item
					finally, if $\theta_i = \Diamond \theta_j$, then $i^{tr} = dseq(i)(j^{tr} \wedge q)$. 
				\end{itemize}
				Then, $\phi^{tr} = k^{tr} \wedge q$ (as $\theta_k$ is actually $\phi$).
				The extra variable, $q$, is used to  mark which prefixes in the $\d_2 \oplus_\subseteq \k$-tableau  correspond to prefixes in the \k-tableau that have appeared.
				
				For convenience assume that in the \k-tableau for $\phi$, $\sigma\ \theta_i$, where $\theta_i = \Diamond\eta$ produces $\sigma.i\ \eta$ -- which is reasonable, since for each $\sigma$ each $\theta_i$ appears at most once.
				Assume a complete accepting \k-branch $b$ for $\phi$. Let 
				$m(0) = 0$ and $m(\sigma.i) = m(\sigma).bin(i)$.
				Then, 
%
%
				$b'$ is constructed in a recursive way, so that for every $\sigma'\ \eta, \sigma'\ q \in b'$, where $\eta \neq q, \neg q$, there is some $\sigma\ \theta_i \in b$ such that $\sigma' = m(\sigma)$ and $\eta = i^{tr}$. 
				When we apply the $\Box_1$- or $\Box_2$-rule from the ones we presented in Table \ref{table:tableauforD2K}, that is in the course of generating a prefix $m(\sigma)$ -- so, from $m(\sigma)\ i^{tr}$, where $\theta_i = \Diamond \theta_j$, we eventually generate $m(\sigma.i)\ j^{tr}$ and $m(\sigma.i)\ q$ (and some auxiliary boxed formulas in-between); when we apply the $\Box_3$-rule, then this started from some $m(\sigma)\ i^{tr}$, where $\theta_i = \Box \theta_j$, so for every $\sigma.l\ \theta_j \in b$, we produce $m(\sigma.l)\ j^{tr}$, while for $\sigma.l\ \theta_j \in b$ (where $l\leq k$), we produce $m(\sigma.l)\ \neg q$ (and  auxiliary boxed formulas in-between); when we apply a propositional rule on $m(\sigma)\ i^{tr} \circ j^{tr}$, we just need to make the same nondeterministic choice that was made for $b$ (if applicable). Then, naturally, if $b'$ is rejecting, then that is because $m(\sigma)\ p, m(\sigma)\ \neg p \in b'$, or $m(\sigma)\ \bot \in b'$; but then either $\sigma\ p, \sigma\ \neg p \in b$, or $\sigma\ \bot \in b$, respectively.
				
				On the other hand it is easier to give a complete accepting \k-branch $b$ for $\phi$ given a complete accepting $\d_2 \oplus_\subseteq \k$-branch $b'$ for $\phi^{tr}$: $b = \{\sigma\ \theta_i \mid m(\sigma)\ i^{tr} \in b' \}$. We leave the reader to verify this claim.
				
				
				Notice that $\chi^{tr}$ has no diamonds and the number of propositional variables in $\chi^{tr}$ is one more than in $\chi$. Since we can assume $\chi$ is variable-free (see \cite{Chagrov02howmany}), the proposition follows.
							\qed
			\end{proof}
			
			For the remaining logics we first present a 
			reduction to show 
			hardness 
			for their diamond-free fragments and then we can use translations to their 1-variable fragments to transfer the lower bounds to these fragments.
							We first treat the case of $\d_2 \oplus_\subseteq \kf$.
			
			\begin{lemma}\label{lem:TMlower}
				The diamond-free fragment of $\d_2 \oplus_\subseteq \kf$ is \EXP-complete, while the diamond-free fragments of $\d \oplus_\subseteq \kf$ and of $\df_2 \oplus_\subseteq \kf$ are \PSPACE-complete.
			\end{lemma}
			
			\begin{proof} 
				The upper bounds were given by Proposition \ref{prp:upper}.
				The proof for the lower bounds resembles the one in \cite{fischer1979propositional} and is by reduction from a generic {\sf A}\PSPACE\ problem given as the alternating Turing machine of two tapes (input and working tape) which uses polynomial space to decide it. Let the  machine be $(Q,\Sigma,\delta,s)$, where $Q$ the set of states, $\Sigma$ the alphabet, $\delta$ the transition relation  and $s$ the initial state. Let $Q = U \cup E$, where $E$ and $U$ are distinct, $E$ the set of existential and $U$ the set of universal states and assume that the machine only has two choices at every step of the computation, provided by two transition \emph{functions}, $\delta_1, \delta_2$: when the transition functions are given state $q \in Q$, and symbols $a,b \in \Sigma$ for tape 1 and 2 respectively, for $i=1,2$, $\delta_i(q,a,b) = (q',c,j_1,j_2) \in Q\times\Sigma\times\{0,-1,1\}^2$, where $q'$ the new state, $c$ the symbol to replace $b$ in tape $2$, and $j_1,j_2$ the respective moves for each tape, where $0$ indicates no move, $-1$ a move to the left, and $1$ a move to the right. 
				Furthermore, let $x = x_1x_2\cdots x_{|x|}$ be the input, where for every $i \in \{1,2,\ldots, |x| \}$, $x_i \in \Sigma$. 
				Since the Turing machine uses polynomial space, there is a polynomial $p$, such that the working tape only uses cells $1$ to $p(|x|)$ for an input $x$. For the input tape, we only need cells $0$ through $|x|+1$ (we may assume additional symbols to indicate the beginning and end of the input), because the head does not go any further and an output tape is not needed, since we are interested only in decision problems. Therefore, there are $Y, N \in Q$, the accepting and rejecting states respectively. Let $r_1 = \{0,1,2,\ldots, |x| + 1 \}$ and $r_2 = \{1,2,\ldots, p(|x|)\}$. A configuration $c$ of the Turing machine is called accepting if the computation of the machine that starts from $c$ is an accepting computation.
				
				For this reduction, a formula will be constructed that will enforce that any model satisfying it must describe a computation by the Turing machine. Each propositional variable will correspond to some fact about a configuration of the machine and the following propositional variables will be used: 
				\begin{itemize}
					\item $t_1[i], t_2[j]$, for every $i \in r_1, j \in r_2$; $t_1[i]$ will correspond to the head for the first tape pointing at cell $i$ and similarly for $t_2[j]$,
					\item $\sigma_1[a,i],\sigma_2[a,j]$, for every $a \in \Sigma$, $i \in r_1, j \in r_2$; $\sigma_1[a,i]$ will correspond to cell $i$ in the first tape having the symbol $a$ and similarly for $\sigma_2[a,j]$ and the second tape,
					\item $q[e]$, for every $e \in Q$; $q[e]$ means the machine is currently in state $e$.
				\end{itemize}
				For each configuration $c$ of the Turing machine there is a formula that describes it. This formula is the conjunction of the following and from now  on it will be denoted as $\phi_c$: $q[e]$, if $e$ is the state of the machine in $c$; $t_1[i]$ and $t_2[j]$, if the first tape's head is on cell $i$ and the second tape's head is on cell $j$; $\sigma_1[a_1,i_1], \sigma_2[a_2,i_2]$, if $i_1 \in r_1, i_2 \in r_2$ and $a_1$ is the current symbol  in cell $i_1$ of the first tape and $a_2$ is the current symbol in cell $i_2$ of the second tape.
					
				We need  the following formulas. Intuitively, a world in a model for $\phi$ corresponds to a configuration of our Turing machine. $q$ ensures there is exactly one state at every configuration; $\sigma$ that there is exactly one symbol at every position of every tape; $t$ that for each tape the head is located at exactly one position; $\sigma'$ ensures that the only symbols that can change from one configuration to the next are the ones located in a position the head points at; $ac$ ensures we never reach a rejecting state (therefore the machine accepts); $st$ starts the computation at the starting configuration of the machine; finally, $d_E, d_U$ ensure for each configuration that the next one is given by the transition relation (functions). Then, if $com = q\wedge \sigma \wedge t \wedge \sigma' \wedge ac \wedge d_{E} \wedge d_{U}$ we define  $\phi = st \wedge com \wedge \Box_3 com$.
				\[
				q = 
				\bigvee_{e \in Q} q[e]
				\wedge \bigwedge_{\substack{e,f \in Q, \\ e\neq f}} \neg \left( q[e] \wedge q[f] \right) ;
				\]
				\[
				\sigma = 
				\bigwedge_{\substack{j \in \{1,2\}, \\ i\in r_j}}\left[
				\bigvee_{a \in \Sigma} \sigma_{j}[a,i]
				\wedge \bigwedge_{\substack{a,b \in \Sigma, \ a\neq b}} \neg \left( \sigma_{j}[a,i] \wedge \sigma_{j}[b,i] \right) \right];
				\]
				\[
				t = 
				\bigwedge_{j \in \{1,2\}} \left[ 
				\bigvee_{i \in r_j} t_j[i] 
				\wedge \bigwedge_{\substack{i,k \in r_j \ i \neq k}} \neg \left( t_j[i] \wedge t_j[k] \right) \right];
				\]
				\[
				\sigma' = 
				\bigwedge_{\substack{j \in \{1,2\},\ \  i, i' \in r_j, \\ i \neq i', \ \ \ a \in \Sigma}} \left[ \left( t_j[i] \wedge \sigma_{j}[a,i'] \right) \rightarrow \Box_1 \sigma_{j}[a,i'] \wedge \Box_2 \sigma_{j}[a,i'] \right];
				\]
				$ac = \neg q[N]$;\\ 
				$	st = \phi_{c_0}$,  where $c_0$ is the initial configuration of the machine;
\\
let 
$locconf(e,i_1,i_2,j_1,j_2) =  
q[e] \wedge \sigma_1[i_1,j_1] \wedge \sigma_2[i_2,j_2] \wedge t_1[j_1] \wedge t_2[j_2]
$ and 
$D(e,k,l_1,l_2,m_1,m_2) = q[e] \wedge \sigma_2[k,l_2] \wedge t_1[l_1+m_1] \wedge t_2[l_2+m_2]$; then,
				\[
				d_{E} = 
				\bigwedge_{\substack{(e,i_1,i_2)\in E\times \Sigma \times \Sigma, \\ j_1 \in r_1, \ j_2 \in r_2}} 
				\left[
				\begin{array}{l} 
				locconf(e,i_1,i_2,j_1,j_2)
				\rightarrow 
				\\ \ \quad
			\begin{array}{l} 
				\Box_1
				D(e_1,k_1,j_1,j_2,m_1^1,m_2^1)
				\\ \
												 \quad
				\vee \
				\Box_1
				D(e_2,k_2,j_1,j_2,m_1^2,m_2^2)
				\end{array}
				\end{array}
				\right] ,
				\]where $(e_1,k_1,m_1^1,m_2^1)= \delta_1(e,i_1,i_2)$, $(e_2,k_2,m_1^2,m_2^2) = \delta_2(e,i_1,i_2)$;
				\[
				d_{U} = 
				\bigwedge_{\substack{(e,i_1,i_2)\in U\times \Sigma \times \Sigma, \\ j_1 \in r_1, \ j_2 \in r_2}} 
				\left[
				\begin{array}{l} 
								locconf(e,i_1,i_2,j_1,j_2)
								\rightarrow 
								\\ \ \qquad 
												\begin{array}{l} 
								\Box_1
								D(e_1,k_1,j_1,j_2,m_1^1,m_2^1)
								\\ \ 
																 \quad
								\wedge \
								\Box_2
								D(e_2,k_2,j_1,j_2,m_1^2,m_2^2)
				\end{array}
				\end{array}
				\right] , 
				\]  where $(e_1,k_1,m_1^1,m_2^1)= \delta_1(e,i_1,i_2)$, $(e_2,k_2,m_1^2,m_2^2) = \delta_2(e,i_1,i_2)$.
				
				The few implications that appear above are of the form $a\wedge b \wedge \cdots \wedge c \rightarrow \psi$ (where $a,b,\ldots,c$ are propositional variables) and can thus be rewritten in negation normal form: $\neg a \vee \neg b \vee \cdots \vee \neg c \vee \psi$.
				 The correctness of the reduction follows from the following two claims.
				
				\emph{Claim: If for some model $\M, w \models \phi$ and for some $u$, such that $(u=w$ or $w R_{3} u)$, $u \models \phi_c$ and $c_1, c_2$ are the next configurations from $c$, then if $c$ a universal configuration, there are $w R_{3} u_1$ and $wR_3  u_2$, such that $u_1 \models \phi_{c_1}$, $u_2 \models \phi_{c_2}$ and if $c$ an existential configuration, there is some $w R_{3} u_1$, such that either $u_1 \models \phi_{c_1}$ or $u_1 \models \phi_{c_2}$.}
				From this claim, it immediately follows that if $\phi$ is satisfiable, then the Turing machine accepts its input (since it never rejects it
				). We prove the claim for the case of the universal configuration. Because of formulas $q, \sigma, t$, in every world $v$, such that $w R_{3} v$, there is exactly one $\phi_c$ satisfied. 
				There are worlds $u_1, u_2$, (because of seriality of $R_{1}, R_{2}$) such that $w R_{1} u_1$ and $w R_{2} u_2$ and if $u_1 \models \phi_{c_3}$, $u_2 \models \phi_{c_4}$, then because of $d_{U}$, $c_3$ will differ from $c$ in all respects $\delta_1$ demands; furthermore, because of $\sigma'$, $c_3$ differs only in the ways $\delta_1$ (or $\delta_2$) demands and we can reason the same way for $c_4$. Therefore, $\{c_3,c_4\}=\{c_1,c_2\}$. 
				
				\emph{Claim: If the Turing machine accepts $x$, then $\phi$ is satisfiable.}
				Given the machine's computation tree for $x$, we can construct model $(W,R_1,R_2,R_3,V)$ for $\phi$. $W$ is the set of configurations in the computation tree; let $R_1,R_2$ be minimal such that if $u$ is a universal configuration and $v,w$ its next configurations, then $u R_1 v$ and $uR_2 w$ (or $uR_2 v$ and $uR_1 w$), while if $u$ an existential configuration and $v$ its next accepting configuration, then $u R_1 v$ and $uR_2 v$; let $R_3$ be the transitive closure of $R_1 \cup R_2$. $V$ is defined to be such that if $\M = (W,R_1,R_2,R_3,V)$, then $\M,u \models \phi_u$. Then, it is not hard to see that $\M,c_0 \models \phi$.
				
				For the case of  $\d \oplus_\subseteq \kf$, notice that if the machine is deterministic, we can eliminate $d_{U}$, half of $d_E$ and the subformulas beginning with $\Box_2$ from $\sigma'$  and rename the remaining modalities from $\Box_1,\Box_3$ to $\Box_1,\Box_2$. 
				For the case of $\df_2 \oplus_\subseteq \kf$, we can define a translation from the language of $\d \oplus_\subseteq \kf$ to the language of $\df_2 \oplus_\subseteq \kf$: given a formula $\phi$ with $\Box_1,\Box_2$ as modalities, 
				simply replace $\Box_2$ by $\Box_1\Box_3\Box_2$ and $\Box_1$ by $\Box_1\Box_2$.
				The remaining argument is similar for the one for the case of $\d_2 \oplus_\subseteq \k$ -- the iteration of $\Box_1$ and $\Box_2$ helps cut off the propagation of boxes in the tableau, which does not happen for $\d \oplus_\subseteq \kf$.
										\qed 
			\end{proof}								
%
%
			From Lemma \ref{lem:TMlower}, with some extra work, we can prove the following.
			
			\begin{proposition} \label{prp:EXPc1var}
				The 1-variable, diamond-free fragment of $\d_2 \oplus_\subseteq \kf$ is \EXP-complete; the 1-variable, diamond-free fragments of $\d \oplus_\subseteq \kf$ and of $\df_2 \oplus_\subseteq \kf$ are \PSPACE-complete.
			\end{proposition}	
			\begin{proof} 
			We  present a method to translate a diamond-free formula $\phi$ in negation normal form into a diamond-free, 1-variable formula $\phi'$ such that $\phi$ is $\d_2 \oplus_\subseteq \kf$-satisfiable iff $\phi'$ is $\d_2 \oplus_\subseteq \kf$-satisfiable. Let 
				$\manyk{p}$ be all the propositional variables that appear in $\phi$ and assume $q$ is not one of them. 
				Then,
				$p_i^v =  \Box_1\Box_2^i q$ and $(\neg p_i)^v =  \Box_1\Box_2^i \neg q$. $\phi'$ results from $\phi$ by replacing each literal $l$ by $l^v$.  Notice that in a model $\M$ and state $u$, only one of $p_i^v$ and $(\neg p_i)^v$ can be true. Let $\M = (W,R_1,R_2,R_3,V)$, where $(W,R_1\cup R_2)$ is an infinite binary rooted tree ($aR_1b$ iff $b$ the left child of $a$ and $aR_2b$ iff $b$ the right child of $a$),  
				$u \in W$, the root, and $\M, u \models \phi$ (it is not hard to see how to construct such a model from any other); $R_3$ is the transitive closure of $R_1 \cup R_2$. Then, for every $x \in W$, if there are some $y \in W$ and some positive $j \in \nat$, such that $y R_1 R_2^j x$ ($R_2^j$ is defined: $R_2^1 = R_2$ and $aR_2^{j+1}b$ iff there is some $c$ s.t. $aR_2cR_2^jb$), then $y,j$ are unique. Thus, if $V'(q) = \{x \in W\mid  \exists y R_1 R_2^j x $ s.t. $y\in V(p_j) \}$, it is the case that for $\M' = (W,R_1,R_2,R_3,V')$, $\M',u \models \phi'$. On the other hand given a model $\M',u \models \phi'$, we can just define $V(p_i) = \{x \in W\mid  \M',x\models \Box_1\Box_2^i q \}$, thus $\phi$ is satisfiable iff $\phi'$ is. 
				If $\phi$ is diamond-free, then $\phi'$ is diamond-free.
				
				
				Notice that the method above does not work for $\d \oplus_\subseteq \kf$. Thus we use another method: we translate a formula $\phi$ to a formula $\phi^1$ such that $\phi$ is $\d \oplus_\subseteq \kf$-satisfiable iff $\phi^1$ is $\d \oplus_\subseteq \kf$-satisfiable and $\phi^1$ only uses one variable. Let $\manyk{p}$ be the propositional variables that appear in $\phi$ and let $q$ be a new variable (not among $\manyk{p}$). Let $s=  q \wedge \Box_1 q \wedge \bigwedge_{i=1}^{k}\Box^{2i+1}\neg q$. Then, we recursively define: 
				$(p_i)^1=\Box_1^{2i} q$;
				$(\neg p_i)^1=\Box_1^{2i} \neg q$;
				$\bot^1=\bot$;
				$(\neg \bot)^1 = \neg \bot$;
				$(\psi_1\wedge \psi_2)^1 = \psi_1^1 \wedge \psi_2^1$;
				$(\psi_1\vee \psi_2)^1 = \psi_1^1 \vee \psi_2^1$;
				$(\Box_1 \psi)^1=\Box_1^{2k+2} (\psi^1 \wedge s)$ ($\Box_1^{x}$ is $x$ iterations of $\Box_{1}$); finally,
				$(\Box_2 \psi)^1= \Box_2 ((\psi^1 \wedge q \wedge \Box_1 q) \vee (\neg q \vee \Box \neg q) )$.
				Formula $s$ gives a ``mold'' to a model. We can assume that the frames for $\d \oplus_\subseteq \kf$ are of the form $(\N,+1,\leq)$. Furthermore, if we restrict ourselves to formulas of the form $\psi^1$, then we can assume that for every $n \in \N$, $n(2k+2), n(2k+2)+1 \models q$ and for $1\leq i\leq k$, $n(2k+2)+2i+1\models \neg q$. Then, $n(2k+2)\models (\Box_1\psi)^1$ if and only if $(n+1)(2k+2)\models \psi^1$, while $q \wedge \Box_1 q$ is true only at multiples of $2k+2$. So,  $n(2k+2)\models (\Box_2 \psi)^1$ exactly when $(n+1)(2k+2)\models \psi^1$. Therefore, by induction on $\phi$, we can see that $\phi$ is $\d \oplus_\subseteq \kf$-satisfiable iff $\phi^1$ is $\d \oplus_\subseteq \kf$-satisfiable.
				
				We can end this argument like the one for the case of $\d_2 \oplus_\subseteq \k$: the tableau run for $\phi$ can simulate the run for $\phi^1$ and vice-versa. 
				Just map every prefix $\sigma$ from the first tableau to $\sigma'$ of the second one, such that $0'=0$ and $(\sigma.1)'=\sigma'.1^{2k+2}$. Then $\sigma\ \psi$ appears in a branch of the first procedure iff $\sigma'\ \psi^1$ appears in a branch which results from the ``same'' nondeterministic choices in the second procedure. Furthermore, it is not hard to see that $\sigma'\ (p_i)^1$ and $\sigma'\ (\neg p_j)^1$ result in a closed branch iff $i=j$.
				
				For the case of $\df_2 \oplus_\subseteq \kf$, simply notice that the translation from $\d \oplus_\subseteq \kf$ in the proof of Lemma \ref{lem:TMlower} does not introduce any variables.
							\qed 
			\end{proof}			
			
					\subsubsection*{Remarks}
			One may wonder whether we can say the same for the variable-free fragment of these logics. The answer however is that we cannot. The models for these logics have accessibility relations that are all serial. This means that any two models are bisimilar when we do not use any propositional variables, thus any satisfiable formula is satisfied everywhere in any model, thus we only need one prefix for our tableau and we can solve satisfiability recursively on $\phi$ in polynomial time.
			
			Notice that for the proofs above, the requirement that the respective accessibility relations are serial was central. Indeed, otherwise there was no way to achieve these results, as we would not be able to force extra worlds in a constructed model. Then we would have to rely on the complexity contributed by propositional reasoning and at best we would get an \NP-hardness result -- as long as we allowed enough variables in our formula.
			
			Then what about $\df \oplus_\subseteq \kf$? Maybe we could  attain similar hardness results for this logic as for $\df_2\oplus_\subseteq \kf$. Again, the answer is no. As frames for $\df$ come with a serial and transitive accessibility relation, frames for $\df \oplus_\subseteq \kf$ are of the form $(W,R_1,R_2)$, where $R_1\subseteq R_2$, $R_1,R_2$ are serial, and $R_1$ is transitive. It is not hard to come up with the following tableau rule(s) for the diamond-free fragment, by adjusting the ones we gave for $\df_2\oplus_\subseteq \kf$ to simply produce $0.1\ \phi$ from every $\sigma\ \Box_i\phi$.
			This drops the complexity of satisfiability for the diamond-free fragment of $\df \oplus_\subseteq \kf$ to \NP\ (and of the diamond-free, 1-variable fragment to \P), as we can only generate two prefixes during the tableau procedure. The following section explores when we can produce hardness results like the ones we gave in this section.
			
			\section{A General Characterization}
			\label{sec:general}
			
			In this section we examine a more general setting and we conclude by establishing 
			tight 
			conditions 
			that determine 
			the complexity of satisfiability of the diamond-free (and 1-variable) fragments of such multimodal logics.

			A general framework would be to describe each logic with a triple $(N,\subset,F)$, where $N=\{1,2,\ldots,|N|\} \neq \emptyset$, $\subset$ a binary relation on $ N$, and for every  $i \in N$, $F(i)$ is a modal logic; a frame for $(N,\subset,F)$ would be $(W,(R_i)_{i\in N})$, where for every $i\in N$, $(W,R_i)$ a frame for $F(i)$ and for every $i\subset j$, $R_i\subset R_j$. It is reasonable to assume that $(N,\subset)$ has no cycles -- otherwise we can collapse all modalities in the cycle to just one -- and that $\subset$ is transitive. Furthermore, we also assume that all $F(i)$'s have frames with serial accessibility relations -- otherwise there is either some $j\subseteq i$ for which $F(j)$'s frames have serial accessibility relations and $R(i)$ would inherit seriality from $R_j$, or when testing for satisfiability, $\Box_i \psi$ can always be assumed true by default (the lack of diamonds means that we do not need to consider any accessible worlds for modality $i$), which allows us to simply ignore all such modalities, making the situation not very interesting from an algorithmic point of view. Thus, we assume that $F(i)\in \{\d,\t, \df, \sv  \}$.\footnote{We can consider more logics as well, but these ones are enough to make the points we need. Besides, it is not hard to extend the reasoning of this section to other logics (ex. {\sf B}, \sr, \kdfv\ and due to the observation above, also \k, \kf), especially since the absence of diamonds makes the situation simpler.}\footnote{Frames for \d\ have serial accessibility relations; frames for \t\ have reflexive accessibility relations; frames for \df\ have serial and transitive accessibility relations; frames for \sv\ have accessibility relations that are equivalence relations (reflexive, symmetric, transitive).} 
			The cases for which $\subset = \emptyset$ have already had the complexity of their diamond-free (and other) fragments determined in \cite{Hemaspaandra2010GeneralizedModalSat}.
			For the general case, we already have an \EXP\ upper bound from \cite{DemriDeNivelle2005decidingregulargram}, as we explain in the following subsection.

			\subsection{Regular Gramar Logics}
			
			We briefly demonstrate that the \EXP\ upper bound from \cite{DemriDeNivelle2005decidingregulargram} applies in the case of $(N,\subset,F)$.
			In this subsection we present the basic definitions about regular grammar logics with converse and we sketch an argument why  $(N,\subset,F)$ is (or can be reduced to) a regular grammar logic with converse. Definitions and most of our arguments come from \cite{DemriDeNivelle2005decidingregulargram}.\footnote{The reader may notice that we give slightly different notation and that we define certain concept differently from \cite{DemriDeNivelle2005decidingregulargram} -- but to the same effect given our purposes.}
			For every agent $i \in N$, let $\overline{i}$ be a new agent, $\overline{\overline{i}} := i$, $\overline{N} := \{\overline{i} \mid i \in N\}$, and in every frame $(W,(R_i)_{i\in N \cup \overline{N}})$, $R_{\overline{i}} = R_i^{-1}$.
			
			A context-free semi-Thue system\footnote{A context-free semi-Thue system is a lot like a context-free grammar, but with no distinction between terminal and non-terminal symbols and no initial symbol. It is a set of (context-free) rules of the form $a \rightarrow \alpha$, where $a \in N \cup \overline{N}$ and $\alpha \in (N \cup \overline{N})^*$. For $\alpha,\beta,\gamma \in (N\cup \overline{N})^*$, $\alpha a \gamma \Rightarrow \alpha\beta\gamma$ if $a \rightarrow \alpha$ is a rule; $\Rightarrow^*$ is the reflexive transitive closure of $\Rightarrow$; we say that $\alpha$ produces $\beta$ if $\alpha \Rightarrow* \beta$.} $g$ on vocabulary $N \cup \overline{N}$ generates a class of frames $c$ where $c$ has all frames such that
			if $i \rightarrow \manyk{i}$ is a rule of $g$, then for every frame $(W,(R_i)_{i\in N \cup \overline{N}}) \in c$, $R_{i_1}R_{i_2}\cdots R_{i_k} \subseteq R_i$.
			
			We correspond a multimodal logic $l$, associated with the class of frames $c$, on agent set $N$ with a context-free semi-Thue system $g_l$ on vocabulary $N \cup \overline{N}$ if $g_l$ generates $c$ and for every rule $i \rightarrow i_1 \cdots i_k$ in $g_l$, there is also rule $\overline{i} \rightarrow \overline{i}_k \cdots \overline{i}_1$ in $g_l$. 
			Then $l$ is called a regular grammar logic with converse if for every $j \in N$, $g_l(j)$, the language produced in $g_l$ from $j$, is regular. 
			
			To argue that $(N,\subset,F)$ is (or can be reduced to) a regular grammar logic with converse, first we examine the possibility that there is an agent $i \in N$, such that $F(i) = \d$ or \df.\footnote{Notice that we consider the general case here, and not only the diamond-free fragment.} We examine the case where there is no agent $F(j) = \df$; if there is, then the reasoning is similar as in the following. If there are such agents, let $N' = N \cup \{i_D\}$, where $i_D$ is a new agent; $F'$ is such that $F'(i) = \kf$ if $i = i_d$, $F'(i) = \k$ if $F(i)=\d$, and $F'(i) = F(i)$ otherwise;  $\subset'$ is defined the same as $\subset$ on $N$ and there is no $i \in N$ such that $i_D \subset' i$, and $i \subset' i_D$ iff $F(i) = \d$. Then we can simply reduce  $(N,\subset,F)$-satisfiability to $(N',\subset',F')$-satisfiability by mapping $\phi$ to $\phi \wedge\Box_{i_D}\bigwedge_{F(i)=\d}\Diamond_{i}\top$.
			
			Then, all single-agent instances of $l =(N,\subset,F)$ are regular grammar logics with converse \cite{DemriDeNivelle2005decidingregulargram}. We show by induction on $m(j) = |\{j \in N \mid j \subset i \}|$ that $g_l(j)$ and $g_l(\overline{j})$ are regular. If $m(j)= 0$, then by the observation above, $g_l(j)$ is regular; $F(j)$ may have any combination of Factivity (given by grammar rule $j \rightarrow \varepsilon$ in $g_l$), Positive Introspection (given by grammar rule $j \rightarrow jj$), and Negative Introspection (given by grammar rule $j \rightarrow \overline{j}j$). The regular language produced by $j$ and $g_l(j)$ is one of $j$, $j+\varepsilon$, $jj^*$, $\overline{j}(j+\overline{j})^* j + j$, $j^*$, $(j+\overline{j})^*$, and  $(j+\overline{j})^* j$, depending on $F(j)$. If $m(j)> 0$, then let $r(j) = \{ g_l(i) \mid i \subset j \}$ and $r(\overline{j}) = \{ g_l(\overline{i}) \mid i \subset j \}$; by the inductive hypothesis, all languages in $r(j)$ and $r(\overline{j})$ are regular. Then, naturally, $g_l(j)$ is the result of replacing $j$ by $(j+ \bigcup r(j))$ and $\overline{j}$ by $(\overline{j}+ \bigcup r(\overline{j}))$ in one of the regular expressions above. The result is a regular expression.
			
			For example, $\d_2 \oplus_\subset \df$ can easily be reduced to $\k_2 \oplus_\subset \kf$ by mapping $\phi$ to $\Diamond_1 \top \wedge \Diamond_2 \top \wedge \Box_3 (\Diamond_1 \top \wedge \Diamond_2 \top) \wedge \phi$ to impose seriality,\footnote{No, this is not the same reduction we described above, but it helps save on notation for the particular example.} for which the corresponding regular languages produced from 1, 2, and 3 would be $\Box_1$, $\Box_2$, and $(\Box_1 + \Box_2 + \Box_3)^*(\Box_1 + \Box_2 + \Box_3)$ respectively. The reader is encouraged to see \cite{DemriDeNivelle2005decidingregulargram} and \cite{nguyen2011exptime} for a more complete presentation of regular grammar modal logics with converse.

\subsection{Characterizing the Complexity of  $(N,\subset,F)$}

We now proceed to characterize the complexity of $(N,\subset,F)$. 
			For every $i \in N$, let \[\min(i) = \{j\in N \mid  j\subset i \mbox{ or } j=i, \mbox{ and }\not\exists j'\subset j \}\] and $\min(N)=\bigcup_{i\in N} \min(i)$. We can now give tableau rules for $(N,\subset,F)$. Let  
			\begin{itemize}
				\item $n_i(\sigma) =\sigma$, if either
				\begin{itemize}
					\item the accessibility relations of the frames for $F(i)$ are reflexive, or
					\item $\sigma = \sigma'.i$ for some $\sigma'$ and the accessibility relations of the frames for $F(i)$ are  transitive;
				\end{itemize}
				\item $n_i(\sigma)=\sigma.i$, otherwise.
			\end{itemize}
			The tableau rules appear in Table \ref{table:generaltableau}.
			\begin{table}[t]
			\begin{minipage}[l][21ex]{0.21\linewidth}
				\[ 
				\inferrule*{\sigma\ \Box_i\phi}{\sigma\ \Box_j \phi}  
				\] where $j\subset i$
			\end{minipage}\hfill
			\begin{minipage}[l][21ex]{0.22\linewidth}
				\[ 
				\inferrule*{\sigma\ \Box_i\phi}{n_i(\sigma)\ \phi}  
				\] where\\ $i\in \min(N)$
			\end{minipage}\hfill
			\begin{minipage}[l][21ex]{0.20\linewidth}
				\[ 
				\inferrule*{\sigma\ \Box_i\phi}{\sigma\ \phi}  
				\] where the frames of $F(i)$ have reflexive acc. relations
			\end{minipage}\hfill
			\begin{minipage}[l][21ex]{0.25\linewidth}
				\[ 
				\inferrule*{\sigma\ \Box_i\phi}{n_j(\sigma)\ \Box_i \phi}  
				\] where $j\in \min(i)$ and $F(i)$'s frames have transitive acc. relations
			\end{minipage}
	\caption{Tableau rules  for the diamond-free fragment of  $(N,\subset,F)$}	
		\label{table:generaltableau}
\end{table}
	
			From these tableau rules we can reestablish \EXP-upper bounds for all of these cases (see the previous sections). To establish correctness, we only show how to construct a model from an accepting branch for $\phi$, as the opposite direction is easier. Let $W$ be the set of all the prefixes that have appeared in the branch. The accessibility relations are defined in the following (recursive) way:  if $i \in \min(N)$, then 
			$ R_i = \{(\sigma,n_i(\sigma)) \in W^2 \} \cup \{(\sigma,\sigma)\in W^2\mid  n_i(\sigma) \notin W $ or $F(i)$ has reflexive frames$ \}; $
			if $i \notin \min(N)$ and  the frames of $F(i)$ do not have transitive or reflexive accessibility relations, then $R_i=\bigcup_{j\subseteq i}R_j$; 
			if $i \notin \min(N)$ and  the frames of $F(i)$ do have transitive (resp. reflexive, resp. transitive and reflexive) accessibility relations, then $R_i$ is the transitive (resp. reflexive, resp. transitive and reflexive) closure of $\bigcup_{j\subseteq i}R_j$. Finally, (as usual) $V(p) = \{w \in W \mid  w\ p $ appears in the branch$ \}$. Again, to show that the constructed model satisfies $\phi$, we use a straightforward induction.
			
%
			
			By taking a careful look at the tableau rules above, we can already make some simple observations about the complexity of the diamond-free fragments of these logics.
			Modalities in $\min(N)$ have an important role when determining the complexity of a diamond-free fragment. In fact, the prefixes that can be produced by the tableau depend directly on $\min(N)$.
			
			\begin{lemma} \label{lem:reflexivenp}
				If for every $i \in \min(N)$, $F(i)$ has frames with reflexive accessibility relations ($F(i)\in \{\t,\sv\}$), then the satisfiability problem for the diamond-free fragment of $(N,\subset,F)$  is \NP-complete and the satisfiability problem for the diamond-free, 1-variable fragment of $(N,\subset,F)$ is in \P.
			\end{lemma}
			\begin{proof} 
				Notice that in this case, for every $i \in \min(N), n_i(0)= 0$, so there is no way to generate any other prefix besides $0$. \NP-hardness is the result of the \NP-hardness of propositional satisfiability. 
				By the above we can restrict ourselves to 1-world models; when we use only one variable, they can all be generated in constant time.
				\qed
			\end{proof}
			
			Taking this reasoning one step further:
			
			\begin{corollary} \label{cor:transitivenp}
				If  $\min(N) \subseteq \{i\}\cup A$ and $F(i)$ has frames with transitive accessibility relations ($F(i)\in \{\df,\sv\}$) and for every $j \in A$, $F(j)$ has frames with reflexive accessibility relations, then the satisfiability problem for the diamond-free fragment of $(N,\subset,F)$  is \NP-complete and the satisfiability problem for the diamond-free, 1-variable fragment of $(N,\subset,F)$ is in \P.
			\end{corollary}
			\begin{proof} 
				Like  Lemma \ref{lem:reflexivenp},  we can only generate  prefixes $0$ and $0.i$. 
				\qed
			\end{proof}
			
			In \cite{DBLP:conf/tableaux/Demri00}, Demri shows that satisfiability for $\L_1\oplus_\subseteq \L_2 \oplus_\subseteq \cdots \oplus_\subseteq \L_n$ is \EXP-complete, as long as there are $i<j\leq n$ for which $\L_i \oplus_\subseteq \L_j$ is \EXP-hard. On the other hand, Corollary \ref{cor:transitivenp} shows that for all these logics, their diamond-free fragment is in \NP, as long as $\L_1$ has frames with transitive (or reflexive) accessibility relations. 
			
			Finally, we can establish general results about the complexity of the diamond-free fragments of these logics. For this, we introduce some terminology.
We call a set $A \subset N$ \emph{pure} if for every $i \in A$, $F(i)$'s frames do not have the condition that their accessibility relation is reflexive (given our assumptions, $F[A] \cap \{ \t,\sv \} = \emptyset$). 
We call a set $A \subset N$ \emph{simple} if for some $i \in A$, $F(i)$'s frames do not have the condition that their accessibility relation is transitive (given our assumptions, $F[A] \cap \{ \d,\t \} \neq \emptyset$). 
An agent $i \in N$ is called pure (resp. simple) if $\{i\}$ is pure (resp. simple).
			
			\begin{theorem}\label{prp:general}
				\begin{enumerate}
					\item If there is some $i\in N$ and some pure $A \subseteq \min(i)$ for which $F(i)$ has frames with transitive accessibility relations ($F(i) \in \{\df,\sv \} $) and either
					\begin{itemize}
						\item  $|A|\ =2$ and $A$ is simple, or
						\item $|A|\ =3$,
					\end{itemize} 
					 then the satisfiability problem for the diamond-free, 1-variable fragment of $(N,\subset,F)$  is \EXP-complete;
					\item 
					otherwise, 
					if there is some $i\in N$ and some pure $A \subseteq \min(i)$ for which
					either
					\begin{itemize}
							\item $|A|\ =2$ and there is some pure and simple $j \in \min(N)$,
or
							\item $|A|\ =3$,
					\end{itemize} 
					then the satisfiability problem for the diamond-free, 1-variable fragment of $(N,\subset,F)$   is 
					\PSPACE-complete;
					\item 
					otherwise, 
					if there is some $i\in N$ and some pure $A \subseteq \min(i)$ for which
					$F(i)$ has frames with transitive accessibility relations ($F(i) \in \{\df,\sv \} $) 
					and either
					\begin{itemize}
							\item $|A|\ =1$ and $A$ is simple or 
							\item $|A|\ =2$, 
					\end{itemize} 
					 then the satisfiability problem for the diamond-free (1-variable) fragment of $(N,\subset,F)$  is 
					\PSPACE-complete;
					\item otherwise the satisfiability problem for the diamond-free (resp. and 1-variable) fragment of $(N,\subset,F)$  is \NP-complete (resp. in \P).
				\end{enumerate}
			\end{theorem}

					\begin{proof} 
						All the lower bounds (except for the one in 4, of course) are established by providing suitable translations. Notice that by definition, if $|\min(i)|>1$, then $i \notin \min(i)$ (and thus, $i\notin \min(N)$). In every case, assume that $\phi$ is the formula that is given.
						
						We first prove 1. This is done by a translation from $\d_2\oplus_\subseteq \kf$. $\phi$ is translated to $\phi^m$, such that $\phi$ is $\d_2\oplus_\subseteq \kf$-satisfiable if and only if $\phi^m$ is $(N,\subset,F)$-satisfiable. 
						If $A=\{x,y\}$ and $F(x)$ has frames with accessibility relations that are not transitive, then let $\Box_{(1)}=\Box_x\Box_x\Box_x\Box_y\Box_x\Box_y\Box_x\Box_x$ and $\Box_{(2)}=\Box_y\Box_x\Box_x\Box_y\Box_x\Box_y\Box_x\Box_x$, and
						$\Box_{(3)}=\Box_i\Box_y\Box_x\Box_y\Box_x\Box_x$. Then, $\phi^{m}$ results from $\phi$ by replacing $\Box_k$ by $\Box_{(k)}$, where $k=1,2,3$.
						We can see that the tableau for $\phi^m$ follows the one for $\phi$ -- as long as we map (say $w$ is mapped to $w^m$) $0$ to $0$ and $\sigma.1$ to $\sigma^m.x.x.x.y.x.y.x.x$ and $\sigma.2$ to $\sigma^m.y.x.x.y.x.y.x.x$. The important observation here is that if $\sigma^m\ \Box_{3}\psi$ eventually produces $\alpha\ \psi$, then $\alpha$ is either some $(\sigma.\tau)^m$, or it is not an initial segment of any such $(\sigma.\tau)^m$. Therefore, by restricting the branches produced by the tableau for $\phi^m$ to prefixes of the form $(\sigma.\tau)^m$, we have a simulation of the corresponding branch for $\phi$, while there are some other prefixes, but we can see that each of those (say $\pi.a$) prefixes a set of formulas, which is a subset of a set of formulas prefixed by another prefix, which ends at $a$ and is mapped from a prefix of the first tableau. This means that if $\phi$ is satisfiable, than $\phi^m$ is satisfiable. The other direction is easier.
						
						When $|A|\ =3$, we can use a more straightforward translation, which resembles than one given to translate from $\d \oplus_\subseteq \kf$ to $\df_2 \oplus_\subseteq \kf$. For $\{a,b,c\}=\{x,y,z\}$, $\phi^{m_a}$ is defined recursively:
						$(p)^{m_a}=p$;
						$(\neg p)^{m_a}= \neg p$;
						$\bot^{m_a}=\bot$;
						$(\neg \bot)^{m_a} = \neg \bot$;
						$(\psi_1\wedge \psi_2)^{m_a} = \psi_1^{m_a} \wedge \psi_2^{m_a}$;
						$(\psi_1\vee \psi_2)^{m_a} = \psi_1^{m_a} \vee \psi_2^{m_a}$;
						$(\Box_1 \psi)^{m_a}=\Box_b \psi^{m_{b}}$, where if $a=x$ (resp. $y$, $z$), then $b=y$ (resp. $z$, $x$);
						$(\Box_2 \psi)^{m_a}=\Box_c \psi^{m_{c}}$, where if $a=x$ (resp. $y$, $z$), then $c=z$ (resp. $x$, $y$); finally,
						$(\Box_3 \psi)^{m_a}= \Box_i \psi$. As the main translation we can pick any of $\phi^{m_x}$,$\phi^{m_y}$,$\phi^{m_z}$ and as in the previous cases, we can simulate one tableau by the other..
						
						To establish the stated lower bound for 2 when $|A|\ =3$, we can use the exact same translation as above (from $\d_2 \oplus_\subseteq \k$). When $|A|\ =2$ and $A=\{x,y\}$, define $\Box_{(1)}=\Box_x\Box_j$, $\Box_{(2)}=\Box_y\Box_j$, and $\Box_{(3)}=\Box_i\Box_j$. The translation happens just by replacing $\Box_a$ by $\Box_{(a)}$, for all $a=1,2,3$. The translations to prove the lower bound for (iii) are just  simplified versions of the above.
						
						To establish the stated upper bounds, we give bounds for the number of prefixes that a tableau run can produce. For this, assume that all subformulas of $\phi$ are distinct.
						
						If for some $\sigma\ \Box_i\psi$, the branch produces both $\sigma.x\ \psi$ and $\sigma.y\ \psi$ (and $x\neq y$), then $x,y \in \min(i)$ and $F(x),F(y)$ have frames with accessibility relations that are not reflexive, which means we are either in case (i) or case (ii). On the other hand, if for $\sigma\ \Box_i\psi$, the branch produces $\sigma.x\ \Box_i\psi$, then $x \in \min(i)$ and $F(x)$ has frames with accessibility relations that are not reflexive, while $F(i)$ has frames with accessibility relations that are transitive. If $F(x)$ has frames with accessibility relations that are not transitive, or there is also some $y \in \min(i)$ such that $F(y)$ has frames with accessibility relations that are not reflexive, then we are either in case (i), or in case (iii). Otherwise, $\sigma$ and $\sigma.x$ are the only prefixes for $\Box_i\psi$ and thus the only possible prefixes for $\psi$ -- if $\sigma.x'$ or $\sigma.x.x'$ is another prefix for $\psi$, then $x' \in \min(i)$, which is a contradiction, because of the above. This establishes 4, because  every subformula of $\phi$ can only have up to 3 prefixes.
						
						Assume we are not in case 1. We give a non-deterministic algorithm which uses polynomial space to solve satisfiability. The algorithm runs the tableau procedure and uses non-determinism exactly for the non-deterministic propositional rule. Since it uses only polynomial space, it (possibly) cannot hold the whole branch in memory, so it explores the prefixes in a certain order.
						This order is what enables the algorithm to use only polynomial space. Every time the algorithm visits prefix $\sigma$, it applies all the rules that have as premise a formula prefixed by $\sigma = \sigma'.y$. This possibly results in new sets of formulas that are prefixed by new prefixes, $\manyk{\sigma.x}$. If there is some $x_a$ such that there is some $i \in N$ for which $x_a,y\in \min(i)$ and $F(i)$ has frames with transitive accessibility relations, then the algorithm visits $\sigma.x_a$ last (there is at most one) and $\sigma.x_a$ is called a \emph{last prefix}. 
						
						If $\sigma.x_b$ is not a last prefix, then the maximum modal depth\footnote{The nesting depth of the boxes in a formula.} of the formulas prefixed by $\sigma.x_b$ is one less than the maximum modal depth of the formulas prefixed by $\sigma'$. This bounds the number of prefixes that are not last prefixes and that are initial segments of a current prefix by at most $2|\phi|$.
						The space the algorithm uses at any time is the number of prefixes it has scheduled to visit (and has not done so yet), times some quantity which is linear with respect to $|\phi|$ (the formulas prefixed by those prefixes). This number of prefixes is at most $|N|$ times the number of initial segments of the current prefix that are not last prefixes. But we argued above that these prefixes are at most $|\phi|$.
						
						Case 4 is given by Corollary \ref{cor:transitivenp}.
									\qed
					\end{proof}
			
			\section{Final Remarks}
			
			We examined the complexity of satisfiability for the diamond-free fragments and the diamond-free, 1-variable fragments of multimodal logics equipped with an inclusion relation $\subset$ on the modalities, such that if $i \subset j$, then in every frame $(W,\manyn{R})$ of the logic, $R_i \subseteq R_j$ (equivalently, $\Box_j \rightarrow \Box_i$ is an axiom). We gave a complete characterization of these cases (Theorem \ref{prp:general}), determining that, depending on $\subset$, every logic falls into one of the following three complexity classes: \NP\ (\P\ for the 1-variable fragments), \PSPACE, and \EXP\ -- Theorem \ref{prp:general} actually distinguishes four possibilities, depending on the way we prove each bound. We argued that to have nontrivial complexity bounds we need to consider logics based on frames with at least serial accessibility relations, which is a notable difference in flavor from the results in \cite{hemaspaandra2001PoorMan,Hemaspaandra2010GeneralizedModalSat}.
			
			One direction to take from here is to consider further syntactic restrictions and Boolean functions in the spirit of \cite{Hemaspaandra2010GeneralizedModalSat}. Another would be to consider different classes of frames. Perhaps it would also make sense to consider different types of natural relations on the modalities and see how these results transfer in a different setting.  
			From a Parameterized Complexity perspective there is a lot to be done, 
			such as limiting the modal depth/width, which are parameters that can remain unaffected from our ban on diamonds. For the cases where the complexity of the diamond-free, 1-variable fragments becomes tractable, a natural next step would be to examine whether we can indeed use the number of diamonds as a parameter for an FPT algorithm to solve satisfiability.
			
			Another direction which interests us is to examine what happens with more/different kinds of relations on the modalities.
			An example would be to introduce the axiom $\Box_i \phi \rightarrow \Box_j\Box_i \phi$, a generalization of Positive Introspection. 
			This would be of interest in the case of the diamond-free fragments of these systems,  as it brings us back to our motivation in studying the complexity of Justification Logic, where such systems exist. Hardness results like the ones we proved in this paper are not hard to transfer in this case, but it seems nontrivial to immediately characterize the complexity of the whole family.
			
			\subsubsection*{Acknowledgments}
			The author is grateful to anonymous reviewers, whose input has greatly enhanced the quality of this paper.

	\end{document}